\documentclass[11pt,a4paper,leqno,oneside]{article}

\usepackage{amsmath}
\usepackage{amsfonts}
\usepackage{amssymb}
\usepackage{amsthm}
\usepackage{makeidx}
\usepackage{graphicx}
\usepackage{pb-diagram}
\usepackage{epstopdf}
\usepackage{fancyhdr}
\usepackage{enumitem}
\usepackage[all]{xy}

\newtheorem{teo}{Theorem}[section]
\newtheorem{prop}[teo]{Proposition} 
\newtheorem{lem}[teo]{Lemma}

\theoremstyle{definition}
\newtheorem{defin}[teo]{Definition}

\newtheorem{rem}[teo]{Remark}

\newtheorem{prob}[teo]{Problem}

   {\begin{verse}%
     \small }%
   {\end{verse}}

\newcommand{\Z}{\mathbb{Z}}
\newcommand{\N}{\mathbb{N}}

\newcommand{\R}{\mathbb{R}}
\newcommand{\ra}{\rightarrow}
\newcommand{\Al}{\mathcal {A}_l}

\title{Optimal strategies for a game on amenable semigroups}

\author{Valerio CAPRARO\thanks{University of Neuchatel, Switzerland. Email: valerio.capraro@unine.ch. This research has been partially supported by Swiss SNF Sinergia project CRSI22-130435.}
\\ Kent E. MORRISON\thanks{California Polytechnic State University, San Luis Obispo, CA; American Institute of Mathematics, Palo Alto, CA. Email: kmorriso@calpoly.edu}}

\begin{document}

\maketitle

\setlength{\parskip}{1ex plus 0.5ex minus 0.2ex}

\begin{abstract}
The semigroup game is a two-person zero-sum game defined on a semigroup $(S,\cdot)$ as follows: Players 1 and 2 choose elements $x\in S$ and $y\in S$, respectively, and player 1 receives a payoff $f(x y)$ defined by a function $f: S \rightarrow [-1,1]$. If the semigroup is amenable in the sense of Day and von Neumann, one can extend the set of classical strategies, namely countably additive probability measures on $S$, to include some finitely additive measures in a natural way. This extended game has a value and the players have optimal strategies. This theorem extends previous results for the multiplication game on a compact group or on the positive integers with a specific payoff. We also prove that the procedure of extending the set of allowed strategies preserves classical solutions: if a semigroup game has a classical solution, this solution solves also the extended game.
\end{abstract}

{\em Keywords:} amenability, multiplicative game, loaded game,
optimal strategies, minimax strategy, Nash equilibrium,
intrinsically measurable \pagebreak \tableofcontents

\section{Introduction}

The \emph{multiplication game} is a two-person zero-sum game in
which the players independently choose positive numbers and multiply
them together. The first player wins when the first digit of the
product is 1, 2, or 3, and the second player wins otherwise.
Different versions of the game arise according to which numbers the
players are allowed to choose. In the original game invented by
Ravikumar \cite{Ra} the players choose from the set of $n$-digit
integers with a fixed $n$. He analyzed the optimal strategies in the
limit as $n \rightarrow \infty$. In \cite{Mo} the second author
showed that when the numbers are positive real numbers, an optimal
strategy for both players is to choose numbers from the Benford
distribution, which is the limit that Ravikumar found in his
analysis.

Also in \cite{Mo} it was shown that the procedure applies naturally
to compact groups. For the \emph{group game} let $G$ be a compact
group and let $W$ be a subset of $G$ that is measurable with
respect to the Haar measure.\footnote{Recall that the Haar measure is the
unique invariant probability measure on a compact group.} The
players choose elements of $G$ and the payoff function of player 1
is defined by
$$
f(x,y)=\left\{
         \begin{array}{ll}
           1, & \hbox{if $xy\in W$} \\
           0, & \hbox{if $xy\notin W$}
         \end{array}
       \right.
$$
The pure strategies are the elements of $G$, which are identified
with the point masses, and mixed strategies are Borel probability
measures on $G$. It turns out that that the Haar measure $\lambda$ is an
optimal strategy for both players, and the value of the game, i.e
the probability that player 1 wins when both players play their own
optimal strategy, is $\lambda(W)$.

Still unsolved are the very natural games using the non-compact group
$(\mathbb Z,+)$ and the non-compact semigroup $(\mathbb N,\cdot)$.
The aim of this article is to extend the result on compact groups to
a larger class of algebraic objects that contains both $(\mathbb
Z,+)$ and $(\mathbb N,\cdot)$. Since neither the existence of
inverses nor an identity element is necessary for playing the
multiplication game, it appears that semigroups are the appropriate
setting for this generalization, which we call the \emph{semigroup
game}. In order to prove the existence of optimal strategies and the
existence of a value for the semigroup game, we restrict our attention to
the class of \emph{amenable} semigroups, but we are forced to
enlarge the set of mixed strategies to include finitely additive
probability measures on $S$. Since no countable group has an invariant countably additive probability measure\footnote{Indeed, the group structure and invariance imply that all singletons have the same measure. }, the result in
\cite{Mo} does not apply to them, but there are countable
groups---for example, abelian groups---for which there are optimal
finitely additive strategies.

Then with the proper interpretation of what a mixed strategy is, we
are able to prove that the game on amenable semigroups has
optimal strategies. This includes many interesting non-compact
groups as well, groups such as the additive group of the integers.

It is worth noting that the second author showed that for the
game on the semigroup $(\mathbb N,\cdot)$ of positive
integers with the multiplication, and with $W=\{\text{integers with
first digit 1 through 3}\}$, there are mixed strategies that are
nearly optimal. Doing this made use of the special structure of $W$
and approximating the game with one on a compact group. This
approach is unfortunately inapplicable to general winning sets.

The first author is grateful to Silvia Ghinassi, Daniel Litt, Vern Paulsen
and Florin R\u adulescu for helpful discussions and to Alain Valette for reading and commenting on an early draft. The authors would like to thank Ted Hill for comments and suggestions. Finally, special thanks go to Marco Dall'Aglio for bringing about the collaboration of the coauthors.

\section{Optimal mixed strategies for the semigroup game}

Although the semigroup game can be defined as long as there is a
binary operation on the set of pure strategies, it seems that
associativity is necessary for our results (see also Remark
\ref{associative}). Recall that a \emph{semigroup} $S$ is a set
equipped with an associative binary operation $S\times S\rightarrow
S$. Given $x,y\in S$, the result of the operation is denoted by
$xy$.

Let $S$ be a semigroup and $f:S\rightarrow [-1,1]$ a function. The semigroup
game $\mathcal G(S,f)$ associated to $S$ and $f$
is the
two-person zero-sum game with $S$ the set of pure strategies for
both players. The payoff function of player 1, which is the negative of the payoff function of player 2, is $f(xy)$. The set of mixed strategies will be specified later.

An interesting case is already when $f$ is the characteristic function of a subset $W$ of $S$. In this case the set $W$ is called \emph{the winning set}.

We first consider the case of a countable semigroup $S$.  (There is a technical
reason to separate the countable case from the
uncountable case.)   Let $L^\infty(S)$ denote the Banach space of
all bounded functions from $S$ to $\mathbb R$, equipped with the supremum norm, and recall
that a \emph{mean} on $S$ is a linear functional
$m:L^\infty(S)\rightarrow\mathbb R$ which is positive, in the sense
that $f\geq0$ implies $m(f)\geq0$, and such that $m(1)=1$. A mean is necessarily bounded and has norm 1. A mean
induces a finitely additive measure on $S$ by defining the measure
of $A\subseteq S$ as $m(\chi_A)$, where $\chi_A$ is the
characteristic function of $A$; we also write
$m(f)=\int_Sf(x)dm(x)$.

Classically, the mixed strategies are probability measures
(countably additive) on the set of pure strategies, with the pure
strategies identified with the point masses $\delta_s$, but for the
semigroup game we expand the mixed strategies to include means,
i.e., finitely additive measures, on $S$. Interest in finitely additive measures has increased in recent years as it has been realized that
\emph{besides technical convenience, there are no
conceptual reasons that support the use of that stronger
assumption} \cite{Ma}. But even early in the development of the rigorous theory of probability it was noted by Kolmogorov \cite[p. 15]{Ko} that \emph{``...in describing any observable random process we can obtain only finite fields of probability. Infinite fields of probability occur only as idealized models of real random processes.''}

It is important to recall that by allowing finitely additive
measures as strategies, the results can be quite
different, as shown by the following classical example.

Consider Wald's game \emph{pick the bigger integer}
: the set of pure strategies is the set of
non-negative integers and the payoff function of player 1 (which is
the negative of the payoff function of player 2) is

$$
f(s,t)=\left\{
         \begin{array}{rl}
           1, & \hbox{if $s>t$} \\
           0, & \hbox{if $s=t$} \\
           -1, & \hbox{if $s<t$}
         \end{array}
       \right.
$$

Wald\cite{Wa} observed that this game has no value if just countably
additive strategies are allowed. In \cite{He-Su} it is shown
that this game has a value if one allows finitely additive
probability measures as strategies, but the value depends on the
order of integration in such a way that the \emph{internal} player
has an advantage\footnote{See \cite{Sc-Se} for more general results
and relation with other phenomena, as de Finetti's
non-conglomerability.}. This seems very strange, since a game
that is naturally symmetric---as presented---becomes asymmetric.
An explanation of this fact will be given in a followup to this
paper \cite{Ca}, where the author shows that Wald's game is equivalent to a semigroup game which is \emph{loadable} (to be defined in the last section) in infinitely many different ways.

For now let us just say that our interpretation of this fact is that
countably additive strategies are \emph{very few} and, on the other
hand, finitely additive strategies are \emph{too many}. This
suggests that the right formulation of the problem should be
somewhere in the middle; i.e., there should be some restrictions on
the set of allowed strategies which lead to a solution of the
problem. So we are now going to propose a natural way to make these
restrictions.

Consider the definition of the payoff to player 1 when he uses the
mixed strategy $p$ and player 2 uses $q$. Assuming that $p$ and $q$
are countably additive probability measures, the payoff $\pi(p,q)$
is the integral of $f(xy)$ with respect to the product
measure $p \times q$ on $S \times S$. By Fubini's Theorem
\[  \pi(p,q)=\int_y\int_xf(xy)\,dp(x)dq(y)=\int_x\int_yf(xy)\,dq(y)dp(x) .\]

But if $p$ and $q$ are only finitely additive, then \emph{the} product measure is not uniquely defined, the
generalization of Fubini's Theorem is not true and the order of
integration matters. If one of the orders of integration is used for
the payoff definition, then there are symmetric games that lose
their symmetry. For example, if the players of Wald's game use
finitely additive mixed strategies, then defining the payoff to be
$\int_y\int_xf(xy)\,dp(x)dq(y)$ gives player 1 the advantage,
whereas changing the order of integration favors player 2
\cite{Sc-Se}.

\begin{defin}
Let $S$ be a countable semigroup. A mean $m$ is \textbf{left-invariant} if
\[  m(f \circ L_s)=m(f) \quad \forall f\in L^\infty(S),s\in S \]
and \textbf{right-invariant} if
\[  m(f \circ R_s)=m(f) \quad \forall f\in L^\infty(S),s\in S \]
where $R_s(x)=xs$ is the right action of $S$ on itself
and $L_s(x)=sx$ is the left action. A mean is \textbf{invariant} if it is both left-invariant and right-invariant.
\end{defin}

Groups or semigroups with invariant means are called
\emph{amenable} and they form an important class of algebraic
objects. The concept of an amenable group was first introduced by J.
von Neumann\cite{vN} and later it was generalized to semigroups
\cite{Da}. An example of an amenable semigroup is the multiplicative
semigroup of natural numbers \cite{Ar-Wi}. Furthermore, not every
group or semigroup is amenable, with the best known example of a
non-amenable group being the free group on two generators\footnote{The free group on two generators, say $x$ and $y$, is the group of all words in the letters $x,x^{-1},y,y^{-1}$, equipped with the operation of concatenation of words, where only the simplifications $xx^{-1}=x^{-1}x=yy^{-1}=y^{-1}y=e$ are allowed, $e$ being the empty word. It was observed by von Neumann that this group, denoted by $\mathbb{F}_2$, is not amenable. A celebrated example of Ol'shanskii shows the existence of non-amenable groups which do not contain $\mathbb{F}_2$ (see \cite{Ol}).}. Likewise,
the free semigroup on two generators is non-amenable. Every finite
group is amenable and the unique invariant mean is given by the
normalized counting measure, but on countably infinite groups there
are no invariant countably additive measures, and so an invariant
mean can only be finitely additive.

From our point of view it is important in the semigroup game to have
the notion of choosing an element ``uniformly'' from $S$. In the
case of a finite group that means the uniform probability measure (which is the unique
invariant mean on a finite group).
In the general setting of semigroups, it is
reasonable to consider an invariant mean as the generalization of
uniform choice. For instance, if we consider the integers with
addition it is intuitively appealing that the probability of
choosing an even number should be the same as the probability of
choosing an odd number and that this probability should be $1/2$.
Indeed any invariant mean on $\mathbb{Z}$ does assign probability
$1/2$ to the even integers and probability 1/2 to the odd
integers.\footnote{The even integers are a 2-\emph{tile}. Given a
semigroup $S$ and a positive integer $k$, possibly infinite. A
subset $W\subseteq S$ is called $k$-tile if there exist
$s_1,...s_k\in S$ such that $S=\bigcup s_iW$ and $s_iW\cap
S_jW=\emptyset$ for $i \neq j$. It is clear that any invariant mean
takes value $\frac{1}{k}$ on a $k$-tile.}

Therefore, in setting up the semigroup game on an amenable semigroup $S$ we fix a particular invariant mean.
\begin{defin}
A \textbf{loading} on $(S,\cdot)$ is given by a finitely additive
probability measure on $S$ which is invariant with respect
to $\cdot$. A loading is denoted by $\ell$.
\end{defin}

We construct the set of allowed strategies following two natural
requirements:
\begin{enumerate}
\item The symmetry and simultaneity of the game suggest that the two players have to be interchangeable, which means that allowed strategies
$p$ and $q$ have to commute in the following sense
$$
\int_y\int_xf(xy)dp(x)dq(y)=\int_x\int_yf(xy)dq(y)dp(x)
$$
This number will be now denoted by $\pi(p,q)$.
Note that if $S$ is commutative\footnote{more generally, if $f(xy)=f(yx)$.}, then this condition is the same as
$\pi(p,q)=\pi(q,p)$.
\item There is no reason to preclude \emph{a priori} a strategy that commutes with all other allowed strategies.
\end{enumerate}

A set of strategies is called \emph{commuting} if the first of the previous two conditions holds.

\begin{defin}
Let $(S,\cdot)$ be loaded with $\ell$, let $f:S\to[-1,1]$, and let $\mathcal A_{\ell}$ be a maximal commuting set of
strategies containing $\ell$. We denote by $\mathcal G(S,f,\mathcal
A_\ell)$ the semigroup game $\mathcal G(S,f)$, when the set of
\emph{allowed strategies} is the set $\mathcal A_\ell$.

Here is a simple lemma, showing that such games exist and admit lots
of strategies.

\begin{lem}\label{lm:allowedstrategies}
For any $f:S\rightarrow[-1,1]$ and for any loading $\ell$, there exists at least one set of allowed strategies $\mathcal A_\ell$
and it contains all the countably additive measures on $S$.
\end{lem}

\begin{proof}
Let $\mathcal F_\ell$ be the family of commuting sets containing
$\ell$ partially ordered by inclusion. First of all this family is
not empty, containing the singleton $\{\ell\}$. Let
$\{C_i\}\subseteq\mathcal F_\ell$ be a chain, the union $\bigcup
C_i$ is easily proved to belong to $\mathcal F_\ell$. It follows
that $\mathcal F_\ell$ is an inductive set and we can apply Zorn's
lemma, getting at least one maximal commuting set $\mathcal A_\ell$
containing $\ell$. To get the second statement, by maximality of $\mathcal A_\ell$, it suffices to show that countably additive measures
commute with all finitely additive measures. Let $p$ be countably
additive and $q$ finitely additive. For a function $f(x,y)$, define
the family of functions $g_x$ by $g_x(y)=f(x,y)$. Now, if we identify $p$ and $q$ with
bounded linear functionals on $L^\infty(S)$, we have that $\int\int f(x,y)dq(y)dp(x)$ is equal to the functional $p$ applied to the function that maps $x$ to the evaluation of $q$ on the function $g_x$. We write this using the following notation:
\[ \int_x \int_y f(x,y) \, dq(y)dp(x) = p(x \mapsto q(g_x)),\]
Since $S$ is countable, $p$ is given by a
non-negative sequence $p_x$ with $\sum_x p_x=1$, and integration
with $p$ is defined by $\int_x \phi(x) \, dp(x) = \sum_x p_x
\phi(x)$. Thus,
\[ p(x \mapsto q(g_x))= \sum_x p_x\, q(g_x) .\] Now $q$ is a bounded linear functional and so the sequence $q(g_x)$ is bounded and that means that $\sum_x p_x q(g_x)$ converges. Therefore \[ \sum_x p_x q(g_x)= q\bigg(\sum_x p_x g_x\bigg), \] which is what we mean by
\[ \int_y \int_x f(x,y) \, dp(x) dq(y) .\]
Thus the order of integration can be switched.
\end{proof}

\end{defin}

We define the payoff to player 1 to be the function
\[ \pi: \Al \times \Al \ra \R : (p,q) \mapsto \int_y \int_x\chi_W(xy)\, dp(x) dq(y) .\]
As usual, we denote by $\underline v$ the lower value of the game, i.e. $\underline v=\sup_p \inf_q \pi(p,q)$, and by $\overline v$ the upper value, i.e. $\overline v= \inf_q \sup_p \pi(p,q)$.

\begin{teo}\label{main}
Let $S$ be an amenable semigroup loaded with $\ell$ and consider the
semigroup game $\mathcal G(S,f,\mathcal A_\ell)$. Then
\begin{enumerate}
\item Both players have an optimal strategy given by any
invariant mean belonging to $\mathcal A_\ell$.
\item The value of the game is $\ell(f)$.
\end{enumerate}
\end{teo}

\begin{proof}
It is certainly true that $\underline v\leq\overline v$, and so we will show that
$\overline v \leq\underline v$. Now by considering what happens when player 1 uses
the invariant mean $\ell$ and player 2 plays an arbitrary $q \in
\mathcal A_\ell$ we see that
\begin{align*}
\underline v=\sup_p\inf_q\pi(p,q) \geq\inf_q\pi(\ell,q)
&=\inf_q\int_y\int_xf(xy)d\ell(x)dq(y)\\
&=\inf_q\int_y\int_xf(x)d\ell(x)dq(y)\\
&=\inf_q\int_y\ell(f)dq(y)\\
&=\ell(f).
\end{align*}
Likewise, when player 2 uses $\ell$ as a strategy and player 1 uses any
strategy $p$ we see that
\begin{align*}
\overline v=\inf_q\sup_p\pi(p,q) \leq\sup_p\pi(p,\ell)
&=\sup_p\int_y\int_xf(xy)dp(x)d\ell(y)\\
&=\sup_p\int_x\int_yf(xy)d\ell(y)dp(x)\\
&=\sup_p\int_x\int_yf(y)d\ell(y)dp(x)\\
&=\sup_p\int_x\ell(f)dp(x)\\
&=\ell(f).
\end{align*}
From these inequalities we see that $\overline v \leq\underline v$, as claimed, and
that $\ell$ is an optimal strategy for both players. Now letting
$m\in\mathcal A_\ell$ be an invariant mean, the same computation as
before shows that the strategy $m$ is optimal for both players.
\end{proof}

This theorem answers in the affirmative the question of the second
author about the existence of optimal strategies in the case of
countably infinite groups as long as the group is
amenable (see the end of \cite{Mo}).\footnote{We underline once again that these strategies are
not defined by $\sigma$-additive measures, but by finitely additive
measures.}

It can be very difficult to prove that a semigroup is amenable
since the invariant means are highly non-constructive objects; the proof of
their existence requires the axiom of choice. So, for instance, it is
hard to answer this easily posed question: \emph{Is the multiplicative
semigroup of natural numbers, $(\mathbb N,\cdot)$, amenable?}
without any other theoretical result. It has been proved in
\cite{Ar-Wi} that $(\mathbb N,\cdot)$ is amenable. Very recently
Vern Paulsen has found an interesting sufficient condition for an
infinite discrete semigroup to be amenable (see \cite{Pa}).

An important question comes immediately to mind. Suppose that a semigroup game has a value $v$ in the classical sense, i.e. with countably additive mixed strategies. Is it always true that for every loading $\ell$ our \emph{extended game} $\mathcal G(S,f,\mathcal A_\ell)$ still has value $v$? The answer is positive. Even more is true: every optimal strategy in the classical sense is still an optimal strategy for the extended game, as shown by the following:

\begin{prop}\label{prop:classicalvsextended}
Let $\sigma$ be an optimal strategy for the semigroup game $\mathcal G(S,f)$ in the classical sense. For any loading $\ell$ and for any strategy set $\mathcal A_\ell$, $\sigma$ is also an optimal strategy for the semigroup game $\mathcal G(S,f,\mathcal A_\ell)$.
\end{prop}

\begin{proof}
The proof uses the fact that countably additive probability measures are dense in the set of all finitely additive strategies with respect to the weak* topology. This is a classical result, whose proof can be found, for instance, in the survey paper \cite[Theorem 4.3]{Me}. Hence, if $\sigma$ is a countably additive probability measure that is optimal in classical sense but not optimal for our extend game, then there exists a finitely additive probability measure $\nu$ such that
$$
\int\int f(xy)d\sigma(y)d\nu(x)=\int\int f(xy)d\nu(x)d\sigma(y)>\int\int f(xy)d\sigma(x)d\sigma(y)
$$
where the first equality follows from Lemma \ref{lm:allowedstrategies}.
Now let $\nu_\alpha$ be a net of countably additive probability measures converging to $\nu$; it follows that for some $\alpha$, one has
$$
\int\int f(xy)d\nu_\alpha(x)d\sigma(y)=\int\int f(xy)d\sigma(y)d\nu_\alpha(x)>\int\int f(xy)d\sigma(x)d\sigma(y)
$$
which contradicts the optimality of $\sigma$ for the classical game.
\end{proof}

We conclude this section with a couple of remarks concerning possible generalizations of our main result.

\begin{rem}
It is possible to prove the analogous result for some uncountable objects. Let $G$ be a locally compact group, which is not compact. Because the Haar measure
is not finite it cannot be used as an invariant mean, but we use the Haar measure to define the Banach space
$$
L^\infty(G)=\{f:G\rightarrow\mathbb R\text{ essentially bounded with
respect to the Haar measure}\}
$$
$G$ is called amenable if there exists an invariant mean on
$L^\infty(G)$ (This is more or less the original definition of J.
von Neumann). The reader can easily write down the statement and the
proof of the analogue of Theorem \ref{main}. However, for the semigroup game on an uncountable semigroup the technical difficulty mentioned earlier is that there is no natural measure to use in the definition of $L^\infty(S)$. Furthermore, it is not clear to the authors whether the second statement of Lemma \ref{lm:allowedstrategies} holds. On the other hand, if Lemma \ref{lm:allowedstrategies} holds, then Proposition \ref{prop:classicalvsextended} also holds, since the weak* density of countably additive probability measures in the set of all finitely additive probability measures is a completely general result that follows from the fact that $L^1(S)$ is weak* dense in its double dual.
\end{rem}

\begin{rem}\label{associative}
The proof of Theorem \ref{main} is independent of the associative property of the operation and so it
shows that whenever a set is equipped with an operation admitting a
finitely additive probability measure that is invariant, then the
operation games are solvable. The point is that without
associativity there are very simple games played on a finite set with a
commutative operation that do not admit invariant finitely additive probability
measures (see \cite{Ca-Da-Sc}).
\end{rem}

\section{Examples and questions}\label{examples}

Consider the additive group of the integers $(\mathbb Z,+)$.
Although it is not compact, it is an amenable group for which
Theorem \ref{main} applies. There is no way to give an explicit
formula for an invariant mean on $\mathbb Z$, but we can build one
from the intuitively appealing concept of density. Let
$A\subseteq\mathbb Z$ and define its density
$$
\mu(A)=\lim_{n\rightarrow\infty}\frac{|A\cap\{-n,-n+1,...-1,0,1,...,n-1,n\}|}{2n+1},
$$
if it exists, where $|X|$ stands for the \emph{number of elements
of $X$}. By the Hahn-Banach theorem, there are invariant means $m$ that extend $\mu$ in the sense that $m(\chi_A)=\mu(A)$ for a set $A$ having a density. Details can be found in \cite[Example 0.3]{Pat}.

Let $W$ be a winning set having a density and load the group
game on $(\Z,+)$ with one of these invariant means. Then an optimal
strategy chooses between odd and even numbers with equal
probability, chooses the last digit with equal probability, etc.
That is, the choice of a congruence class mod $k$ should be done
with probability $1/k$.

The most natural setting for the original multiplication game is the
multiplicative semigroup of positive integers $(\mathbb N,\cdot)$;
that is, the players each choose a positive integer without any
further restrictions. This semigroup is amenable and so Theorem
\ref{main} applies. Any invariant mean apparently exhibits some
strange characteristics: for any set of the form $k\N$ the measure
is 1 because it is the same as the measure of $\N$, and thus on the
complement of $k\N$ the measure is 0. This means that the even
numbers have measure 1 and the odd numbers have measure 0. A player
whose winning set is the even numbers is sure to win, which in fact
makes sense, because he or she can choose an even number to win. As
in the previous example with $(\mathbb Z,+)$, we can construct an
invariant mean which behaves like a density. Let indeed $P_n$ be the
set of natural numbers whose prime factorization contains just the
first $n$ primes, each of them with power at most $n$. Let
$A\subseteq \N$, define its multiplicative density to be

$$
\mu(A)=\lim_{n\rightarrow\infty}\frac{|A\cap P_n|}{|P_n|},
$$
if it exists. As before, there are invariant means which extend
$\mu$.

It has been already observed that in the case of compact groups there is
a unique invariant countably additive probability measure and so the
game is implicitly loaded with such a measure. Now when we consider
the larger sets of mixed strategies that are finitely additive, we
generally lose the uniqueness.\footnote{In very special cases there
are unique invariant means. See, for instance, chapter 7 of the book
\cite{dlH-Va} and references therein for a treatment of the so-called Ruziewicz
problem.} This leads us to the following question:

\begin{itemize}
\item Does there exist a semigroup game which is loadable in
infinitely many different ways?
\end{itemize}

This question is very important in our opinion, since a
positive answer would imply that loadings are really in some sense
\emph{part of the rules of the game}: some games cannot be played in
a coherent way without fixing a loading a priori.

In order to answer this question we make the following definition.

\begin{defin}
Let $S$ be a countable amenable semigroup and $f:S\rightarrow [-1,1]$ be a bounded function. We introduce the following two numbers
$$
f^{-}=\inf\{m(f) | \text{ $m$ invariant mean}\}
$$
and
$$
f^+=\sup\{m(f) | \text{ $m$ invariant mean}\}
$$
We say that $f$ has \textbf{property IM} (Intrinsic Measurability) if
$f^-=f^+$.
\end{defin}

A set is said to have the property IM if its characteristic function has the property IM. For example, any tile has the property IM\footnote{For explicit
examples consider $(\mathbb Z,+)$ for which any congruence class is a
tile.}. Moreover, the class of sets with IM is closed under the
following two operations: if $A,B$ have the property IM and they are
disjoint, then $A\cup B$ has the property IM\footnote{It is false
that $A\cup B$ is IM, when $A$ and $B$ intersect.}; if $A,B$ have the property IM and $A\subseteq
B$, then the difference $B\setminus A$ has the
property IM. It would be nice, in relation to the earlier work of
the second author \cite{Mo}, to know whether or not the set of
positive integers with first digit 1,2,3 has the property IM with
respect to the multiplication. An explicit example of sets without
IM will be given in course of proof of the
following

\begin{prop}\label{strangexamples}
The previous question has a positive answer; namely, there is a semigroup game with uncountably many different loadings.
\end{prop}

\begin{proof}
Let $S$ be the additive semigroup of integers which are greater than
or equal $2$ and $W=\{n\in\mathbb
N:\lfloor\log_2\log_2(n)\rfloor\text{ is even}\}$. We now prove that $W^-=0$ and $W^+=1$. First of all, observe that we can re-write $W$ in the following form
$$
W=\bigcup_{k=0}^\infty[2^{2^{2k}},2^{2^{2k+1}}-1]
$$
where $[a,b]$ stands, in this case, for the set of integers $x$ such that $a\leq x\leq b$. Now, take $n$ of the form $2^{2^{2j}}-1$, for some $j$, and observe that
$$
\frac{|W\cap[2,n]|}{n-2}\leq\frac{2^{2^{2(j-1)}}}{2^{2^{2j}}-3}
$$
which goes to $0$, when $j$ goes to infinity. It follows that
$$
\lim\inf_n\frac{|W\cap[2,n]|}{n-2}=0
$$
In a similar way, choosing $n$ of the form $2^{2^{2j+1}}$, one gets

$$
\lim\sup_n\frac{|W\cap[2,n]|}{n-2}=1
$$

Now,  the set of values which are taken by some invariant
mean over $(\mathbb N,+)$ is convex (since convex combinations of invariant means are still invariant means) and contains the $\lim\inf$ and the $\lim\sup$ above (this is a standard fact and a proof can be found in \cite{Be}). It
follows that for any $r\in[0,1]$ there is an invariant mean $\ell$
such that $\ell(W)=r$. These invariant means give an uncountable
family of different loadings.
\end{proof}

Another important question that comes to mind regards the possibility to associate a well defined value to games to which one could not associate a value up to now: does there exist a semigroup game with $S$ and $f$ that has no value in the classical sense, whose value depends on the order of integration when the mixed strategies are all finitely additive probability measures, but $f$ has the property IM? More formally, the question is

\begin{prob}
Does there exist an amenable semigroup $S$ and an IM function $f:S\rightarrow[-1,1]$, such that:
\begin{itemize}
\item $\sup\inf\pi(p,q)<\inf\sup\pi(p,q)$, where $p$ and $q$ range over all countably additive probability measures on $S$,
\item $\sup\inf\int\int f(xy)dp(x)dq(y)\neq\sup\inf\int\int f(xy)dq(y)dp(x)$, where $p$ and $q$ range over all finitely additive probability measures.
\end{itemize}
\end{prob}

We do not know the answer to this question. Indeed, the first idea is to construct an IM function which does not verify Fubini's property\footnote{We say that a bounded function $f:S\times S\rightarrow\mathbb R$ has Fubini's property if and only if $\int\int f(x,y)d\mu(x)d\nu(y)=\int\int f(x,y)d\nu(y)d\mu(x)$, for all finitely additive probability measures $\mu,\nu$ on $S$.} and this is easy, since the function $f:\mathbb Z\rightarrow[-1,1]$, defined by $f(x)=\chi_{2\mathbb N}(x)-\chi_{2\mathbb N+1}(x)$, already plays the role\footnote{$f$ has clearly the property IM. To see that it does not verify Fubini's property it suffices to take $\mu$ to be the trivial extension to $\mathbb Z$ of an invariant probability measure on the additive semigroup $2\mathbb N$ and $\nu$ to be the \emph{inversion} of $\mu$; i.e. $\nu(A)=\mu(-A)$, for all $A\subseteq\mathbb Z$.}. But this is not enough to exhibit an example as required, since the group game $\mathcal G(S,f)$ has value $0$ even when all finitely additive measures are allowed. It is indeed quite possible that such an example does not exist and that our extension procedure turns to be equivalent to choosing each order of integration with some probability. We believe that this latter possibility is intriguing, since this approach to solving infinite games (deciding one of the two orders of integration with some probability) was analyzed in \cite{Sc-Se}, where the authors proved in their Theorem 2.4 that, under the so-called condition A, the game has a solution in some metric completion of the set of all finitely additive strategies. So it would be nice to discover that the two approaches are actually equivalent for the semigroup game (and, in particular, that one can find a solution in the set of finitely additive measures without passing to some metric completion). In fact, the approach of Schervish and Seidenfeld does not require the axiom of choice, and so it is apparently more realizable\footnote{\emph{Apparently} means that, in fact, even when the Schervish-Seidenfeld theorem gives a solution in the set of all finitely additive strategies, in most cases, this solution is a purely finitely additive measures, and it has been recently proved by Lauwers that such measures are non-constructible objects as well (see \cite{La}).}, but from a theoretical point of view it fails to capture the essence of the problem, which, in case of the semigroup game, is the lack of uniqueness of the loading. In order to make this observation clearer, consider that when we play a game with \emph{fair} dice, we expect that the probability is $\frac16$ for each face, and this is in fact the unique loading on a set of six elements. The motivation of our research is to find the counterpart for infinite sets of this procedure, and we have used the notion widely accepted among group theorists that invariant means are the infinite analogue of uniform measures. But the lack of uniqueness of an invariant mean reflects the ambiguity of the word \emph{fair};  to fix a loading (i.e., to choose an invariant mean) is to define what fairness is.

\end{document}